\documentclass[peerreview,a4paper,12pt]{IEEEtran}

\usepackage{amsthm}
\usepackage{amsmath}
\usepackage{amsfonts,amssymb,amsmath}
\usepackage{graphicx}
\usepackage{epsfig}
\usepackage[applemac]{inputenc}
\usepackage{latexsym}
\usepackage{dsfont}

\newtheorem{mydef}{Definition}
\newtheorem{mycor}{Corollary}
\newtheorem{myprop}{Proposition}
\newtheorem{mythe}{Theorem}

\newtheorem{myrem}{Remark}

\newtheorem{myex}{Example}

\DeclareMathOperator*{\LLR}{LLR}
\DeclareMathOperator*{\BSC}{BSC}
\DeclareMathOperator*{\BEC}{BEC}
\DeclareMathOperator*{\GP}{GP}

\DeclareMathOperator*{\SCE}{SCE}

\DeclareMathOperator*{\bic}{bic}

\begin{document}

\sloppy

\title{Erasure Schemes Using Generalized Polar Codes: Zero-Undetected-Error Capacity and Performance Trade-offs}

\author{
  \IEEEauthorblockN{Rajai Nasser\\}
  \IEEEauthorblockA{
Ecole Polytechnique F\'{e}d\'{e}rale de Lausanne,\\
Lausanne, Switzerland\\
Email: rajai.nasser@epfl.ch} 
}



\maketitle

\begin{abstract}
We study the performance of generalized polar (GP) codes when they are used for coding schemes involving erasure. GP codes are a family of codes which contains, among others, the standard polar codes of Ar{\i}kan and Reed-Muller codes. We derive a closed formula for the zero-undetected-error capacity $I_0^{\GP}(W)$ of GP codes for a given binary memoryless symmetric (BMS) channel $W$ under the low complexity successive cancellation decoder with erasure. We show that for every $R<I_0^{\GP}(W)$, there exists a generalized polar code of blocklength $N$ and of rate at least $R$ where the undetected-error probability is zero and the erasure probability is less than $2^{-N^{\frac{1}{2}-\epsilon}}$. On the other hand, for any GP code of rate $I_0^{\GP}(W)<R<I(W)$ and blocklength $N$, the undetected error probability cannot be made less than $2^{-N^{\frac{1}{2}+\epsilon}}$ unless the erasure probability is close to $1$.
\end{abstract}

\section{Introduction}

Polar coding, invented by Ar{\i}kan \cite{Arikan}, is the first low complexity coding technique that achieves the symmetric capacity of binary-input memoryless channels. Polar codes rely on a phenomenon called \emph{polarization}, which is the process of converting a set of identical copies of a given single user binary-input channel, into a set of ``almost extremal channels", i.e., either ``almost perfect channels'', or ``almost useless channels".

The invention of polar codes brought attention to Reed-Muller codes because of their similarity. It was recently shown that Reed-Muller codes achieve the capacity of binary erasure channels under MAP decoding \cite{RMCodes}.

The probability of error of polar codes under successive cancellation decoding  was shown to be equal to $o(2^{-N^{1/2-\epsilon}})$ by Ar{\i}kan and Telatar \cite{ArikanTelatar}. A more refined estimation of the probability of error (which is dependent on the transmission rate $R$) was obtained by Hassani et al. \cite{HasUrb}. They showed that the probability of error under successive cancellation decoding of the polar code is equal to $2^{-2^{\frac{n}{2}+\frac{\sqrt{n}}{2}Q^{-1}\left(\frac{R}{I(W)}\right)+o(\sqrt{n})}}$ where $N=2^{n}$ is the blocklength, $R$ is the transmission rate and $I(W)$ is the capacity of the binary memoryless symmetric (BMS) channel $W$. They also showed that the probability of error under MAP decoding has the same asymptotic estimation. This does not show a good performance of polar codes in terms of the probability of error because the decay is too slow in the blocklength. One attempt to enhance the performance of polar codes was to apply list decoding with CRC error detection \cite{TalVardy}.

Another possible way to enhance the performance of polar codes is through decoding with erasure; it is sometimes desirable to allow the receiver not to decide which message was transmitted, especially when there is a feedback from the receiver to the transmitter: If a confusing string of symbols was received (in the sense that there is a high probability of a decoding error to occur, no matter which message the receiver chooses as the decoded message), the receiver can ask the transmitter to retransmit the message, hoping that the received string will not be confusing in the next transmission.

There are two types of error in decoding with erasure:
\begin{itemize}
\item If the receiver decides on the transmitted message and makes an error, we say that an undetected error occurs.
\item If the receiver does not decide, we say that an erasure occurs.
\end{itemize}
In general, there is a trade-off between the probability of undetected error $p_{\mathrm{ue}}$ and the erasure probability $p_{\mathrm{er}}$: $p_{\mathrm{ue}}$ can be made smaller at the expense of a higher $p_{\mathrm{er}}$. The trade-off between these parameters was first studied by Forney \cite{Forney}. In this paper, we study the tradeoff between these parameters for generalized polar (GP) codes, which are a family of codes that contains, among others, the standard polar codes of Ar{\i}kan and Reed-Muller codes. Moreover, we compute the zero-undetected-error capacity of GP codes under the low complexity successive cancellation decoder with erasure. We also derive an estimate of the erasure probability of GP codes for rates which are less than the zero-undetected-error capacity.
\section{Preliminaries}

\subsection{Useful notations}

For $0\leq \epsilon,\epsilon'\leq 1$, define the following:
\begin{itemize}
\item $\overline{\epsilon}=1-\epsilon$.
\item $\epsilon\ast\epsilon'=\epsilon\overline{\epsilon'}+\overline{\epsilon}\epsilon'$.
\item $m(\epsilon)=\min\{\epsilon,\overline{\epsilon}\}$.
\end{itemize}

For every $x\in\mathbb{F}_2^N$ and every $\mathcal{I}\subset [N]=\{1,\ldots,N\}$, we write $x_{\mathcal{I}}\in\mathbb{F}_2^{\mathcal{I}}$ to denote the subvector containing the components of $x$ whose indices appear in $\mathcal{I}$.

\subsection{Erasure Schemes}

Let $W:\mathbb{F}_2\longrightarrow\mathcal{Y}$ be a binary input channel. A coding scheme with erasure is a triple $\mathbf{C}=(\mathcal{M},f,g)$ where $\mathcal{M}$ is the set of messages, $f:\mathcal{M}\to\mathbb{F}_2^N$ is the encoder mapping, $N$ is the blocklength of the code, $g:\mathcal{Y}^N\to\mathcal{M}\cup\{\mathbf{e}\}$ is the decoder mapping and $\mathbf{e}\notin\mathcal{M}$ represents erasure.

The scheme is used as follows:
\begin{itemize}
\item The transmitter chooses a message $\mathbf{m}$ uniformly in $\mathcal{M}$ and computes $X^N=(X_1,\ldots,X_N)=f(\mathbf{m})$. 
\item The transmitter sends $X_1,\ldots,X_N$ through $N$ independent copies of the channel $W$, i.e., he uses the channel $N$ times. The rate $R$ of the coding scheme is the amount of information that is sent per channel use: $\displaystyle R=\frac{\log_2|\mathcal{M}|}{N}$.
\item The receiver obtains $Y_1,\ldots,Y_N$ and computes $\hat{\mathbf{m}}=g(Y^N)=g(Y_1,\ldots,Y_N)$.
\item If $\hat{\mathbf{m}}=\mathbf{e}$, we say that an erasure has occurred. Thus, the erasure probability of the scheme is
$p_{er}(W,\mathbf{C})=\mathbb{P}(\{\hat{\mathbf{m}}=\mathbf{e}\})$.
\item If $\hat{\mathbf{m}}\neq\mathbf{e}$ and $\hat{\mathbf{m}}\neq\mathbf{m}$, we say that an undetected error has occured.	Therefore, the undetected error probability of the scheme is
$p_{ue}(W,\mathbf{C})=\mathbb{P}\big(\big\{\hat{\mathbf{m}}\notin\{\mathbf{e},\mathbf{m}\}\big\}\big)$.
\end{itemize}

In practice, it is desirable to maximize the rate $R$ while minimizing the blocklength $N$, the erasure probability $p_{er}(W,\mathbf{C})$, the undetected-error probability $p_{ue}(W,\mathbf{C})$ as well as the computational complexity of both the encoder and the decoder. The trade-off between all these performance parameters is one of the most important problems in information theory. In this paper we are interested in studying the trade-off between these parameters asymptotically in $N$ under the following assumptions:
\begin{enumerate}
\item[i] A BMS channel $W$ is used.
\item[ii] Only GP codes are considered.
\item[iii] Only successive cancellation decoders with erasure are considered.
\end{enumerate}

\subsection{Binary-input memoryless symmetric channels}

Binary-input memoryless symmetric (BMS) channels generalize binary symmetric channels (BSC). One can think of a BMS channel as ``a combination of BSCs": Let $\BSC(\epsilon_1),\ldots,\BSC(\epsilon_l)$ be a collection of $l$ binary symmetric channels of crossover probabilities $\epsilon_1,\ldots,\epsilon_l$ respectively. Let $p_1,\ldots,p_l$ be a probability distribution over $[l]:=\{1,\ldots,l\}$ and consider the binary input channel $W$ which operates as follows: During each use of the channel $W$, one of the channels $\BSC(\epsilon_1),\ldots,\BSC(\epsilon_l)$ is chosen with probability $p_1,\ldots,p_l$ respectively. The bit at the input of $W$ is transmitted to the receiver through the chosen BSC. Moreover, we assume that the receiver knows which BSC was used in each channel use of $W$. Formally, the channel $W:\mathbb{F}_2\to[l]\times\mathbb{F}_2$ can be defined as follows:
\begin{equation}
W(i,y|x)=\begin{cases}p_i\cdot(1-\epsilon_i)\quad&\text{if}\;x=y,\\p_i\cdot\epsilon_i\quad&\text{if}\;x\neq y.\end{cases}
\label{eqBMS}
\end{equation}
We denote this channel $W$ as $$\displaystyle W=\sum_{i=1}^l p_i\cdot\BSC(\epsilon_i).$$

\begin{mydef}
A channel $W$ is said to be binary-input memoryless symmetric (BMS) if there exist $0\leq\epsilon_1,\ldots,\epsilon_l\leq 1$ and a probability distribution $\{p_1,\ldots,p_l\}$ over $[l]=\{1,\ldots,l\}$ such that $W$ is equivalent to (in the sense that it is both upgraded and downgraded from) the channel $\displaystyle\sum_{i=1}^l p_i\cdot\BSC(\epsilon_i)$. In this case, we write 
\begin{equation}
W\equiv \displaystyle\sum_{i=1}^l p_i\cdot\BSC(\epsilon_i),
\label{eqCanon}
\end{equation}
and we say that this is a $\BSC$-decomposition of $W$.
\end{mydef}

Note that one can define general BMS channels by considering infinite collections of BSCs. The binary-input additive white Gaussian noise channels are examples of general BMS channels with continuous output alphabet. For the sake of simplicity, we will only consider in this paper BMS channels with finite output alphabets. However, all the main results of this paper are also valid for general BMS channels.

Another remark worth mentioning is that there are infinitely many $\BSC$-decompositions of a given BMS channel $W$. The reason for this is twofold:
\begin{enumerate}
\item[(i)] We can decompose or unite $\BSC$-components having the same crossover probability by decomposing or adding their fractions (i.e., the $p_i$ parameters) respectively.
\item[(ii)] For every $\epsilon>0$, we have $\BSC(\epsilon)\equiv\BSC(\overline{\epsilon})$, therefore we can change the crossover probability of any BSC component to its complement.
\end{enumerate}

This motivates the following definition:

\begin{mydef}
If $\epsilon_i\leq \frac{1}{2}$ for all $1\leq i\leq l$, we say that $\displaystyle W\equiv\sum_{i=1}^l p_i\cdot\BSC(\epsilon_i)$ is a natural $\BSC$-decomposition of $W$. Note that any $\BSC$-decomposition can be naturalized as follows:
$$\displaystyle W\equiv\sum_{i=1}^l p_i\cdot\BSC(\epsilon_i)\equiv \sum_{i=1}^l p_i\cdot\BSC\big(m(\epsilon_i)\big).$$

If $0\leq\epsilon_1<\ldots<\epsilon_l\leq \frac{1}{2}$ and $p_i>0$ for all $1\leq i\leq l$, we say that $\displaystyle W\equiv\sum_{i=1}^l p_i\cdot\BSC(\epsilon_i)$ is the canonical $\BSC$-decomposition of $W$. It can be shown that the canonical $\BSC$-decomposition of $W$ is unique.
\end{mydef}

\begin{myex}
For every $0\leq\epsilon\leq 1$, the binary erasure channel $\BEC(\epsilon)$ is BMS. Moreover, its canonical $\BSC$-decomposition is $$\BEC(\epsilon)\equiv(1-\epsilon)\cdot \BSC(0)+\epsilon \cdot \BSC\left(\frac{1}{2}\right).$$
\end{myex}

\begin{mydef}
\label{defpw}
Let $\displaystyle W\equiv \sum_{i=1}^l p_i\cdot\BSC(\epsilon_i)$. For every $0\leq\epsilon\leq \frac{1}{2}$, define the fraction $p_W(\epsilon)$ of $\BSC(\epsilon)$ in $W$ as follows:
\begin{equation*}
p_W(\epsilon)=\sum_{i=1}^l p_i\cdot \mathds{1}_{\{m(\epsilon_i)=\epsilon\}}.
\end{equation*}
$p_W(\epsilon)$ is well defined because it does not depend on the $\BSC$-decomposition of $W$. I.e., if $\displaystyle \sum_{i=1}^l p_i\cdot\BSC(\epsilon_i)\equiv \sum_{j=1}^{l'} p_j'\cdot\BSC(\epsilon_j')$ then $\displaystyle\sum_{i=1}^l p_i\cdot \mathds{1}_{\{m(\epsilon_i)=\epsilon\}}=\sum_{j=1}^{l'} p_j\cdot \mathds{1}_{\{m(\epsilon_j')=\epsilon\}}$.
\end{mydef}


As we will see later, the parameter $p_W(0)$ will play an important role in our analysis. We introduce another parameter which is also of interest for our study:

\begin{mydef}
Let $W$ be a BMS channel. We define the \emph{best imperfect component} of $W$, denoted $\epsilon_{\bic}(W)$, as follows:
\begin{align*}
\epsilon_{\bic}(W)&=\begin{cases}
0\quad&\text{if}\;I(W)=1,\\
\displaystyle\min_{\substack{\epsilon\in ]0,\frac{1}{2}]:\\p_W(\epsilon)>0}} \epsilon\quad&\text{if}\;I(W)<1,
\end{cases}
\\
&=
\begin{cases}
0\quad&\text{if}\;I(W)=1,\\
\displaystyle\min_{\substack{1\leq i\leq l,\\p_i>0,\;0<\epsilon_i<1}}m(\epsilon_i)\quad&\text{if}\;I(W)<1,
\end{cases}
\end{align*}
\end{mydef}


\subsection{$\mathcal{D}_t$ decoders for BMS channels}

\begin{mydef}
Let $W=\sum_{i=1}^l p_i\cdot\BSC(\epsilon_i)$ and let $0\leq t\leq \frac{1}{2}$. Define the decoder $\mathcal{D}_t:[l]\times\mathbb{F}_2\to \{0,1,\mathbf{e}\}$ of $W$ as follows:
$$\mathcal{D}_t(i,x)=\begin{cases}x\quad&\text{if}\;\epsilon_i\leq t,\\
1\oplus x\quad&\text{if}\;\epsilon_i\geq 1-t,\\
\mathbf{e}\quad&\text{otherwise}.\end{cases}$$
\end{mydef}

\begin{myrem}
$\mathcal{D}_t$ decoders are desirable because no other decoder with erasure can provide a strictly better trade-off between $p_{ue}$ and $p_{er}$ for the code of blocklength 1 and rate 1. Moreover, $\mathcal{D}_t$ decoders are very easy to implement: we compute the log-likelihood ratio $\LLR(y)=\log\frac{\mathbb{P}_{X|Y}(1|y)}{\mathbb{P}_{X|Y}(0|y)}$ (where $X$ and $Y$ are the input and output of $W$ respectively) and then compare with $T=\log\frac{1-t}{t}$:
$$\mathcal{D}_t(y)=\begin{cases}0\quad&\text{if}\;\LLR(y)\leq -T,\\
1\quad&\text{if}\;\LLR(y)\geq T,\\
\mathbf{e}\quad&\text{otherwise}. \end{cases}$$
\end{myrem}

\subsection{Generalized polar codes}

\begin{mydef}
A code $f:\mathcal{M}\to\mathbb{F}_2^N$ is said to be a \emph{generalized polar (GP) code} of parameters $(n,r,\mathcal{I},b)$ if it satisfies the following:
\begin{itemize}
\item $N=2^n$, $\mathcal{M}=\mathbb{F}_2^{\mathcal{I}}$ and $b\in \mathbb{F}_2^{N-r}$.
\item $\mathcal{I}\subset [N]=\{1,\ldots,N\}$ and $|\mathcal{I}|=r$.
\item $f(u)=F^{\otimes n}\cdot \tilde{u}$, where
$$F=\begin{bmatrix}
1 & 1\\
0 & 1
\end{bmatrix},$$
and $\tilde{u}\in\mathbb{F}_2^N$ is such that $\tilde{u}_\mathcal{I}=u$ and $\tilde{u}_{\mathcal{I}^c}=b$.
\end{itemize}
$n$ is called the \emph{number of polarization steps} of the GP code. We denote the code $f$ as $\GP(n,r,\mathcal{I},b)$. Moreover, if $b=0\in\mathbb{F}_2^{N-r}$, we simply write $\GP(n,r,\mathcal{I})$.
\end{mydef}

\begin{myex}
Here are two examples of GP codes:
\begin{itemize}
\item Standard polar codes of Ar{\i}kan: Take $\mathcal{I}$ to be the set of indices of the $r$ synthetic channels having the lowest Bhattacharyya parameters, and take $b$ to be the vector of frozen bits.
\item Reed-Muller codes: Take $\mathcal{I}$ to be the set of indices of the $r$ columns of $F^{\otimes n}$ having the largest number of ones, and take $b=0\in\mathbb{F}_2^{N-r}$.
\end{itemize}
\end{myex}

\subsection{Successive cancellation decoder with erasure of GP codes}

Because of the recursive construction of $F^{\otimes n}$, one can implement the encoder of any GP code in $O(N\log N)$ time exactly like polar codes.

On the other hand, for any given $\GP(n,r,\mathcal{I},b)$ code, there are various decoders that can be considered. One attractive choice is what we call \emph{successive cancellation decoder with erasure (SCE)} which operates similarly like the successive cancellation decoder of polar codes, but instead of applying the ML decoder for each bit $u_i$, we apply a $\mathcal{D}_{t_i}$ decoder for some $0\leq t_i\leq \frac{1}{2}$. The reason why SCE decoders are desirable is because they have low computational complexity.

\begin{mydef}
For every $i\in\mathcal{I}$ let $0\leq t_i\leq\frac{1}{2}$ and let $t=(t_i)_{i\in\mathcal{I}}\in [0,\frac{1}{2}]^{\mathcal{I}}$. The $\mathcal{D}_t$ successive cancellation decoder with erasure (denoted $\SCE\text{-}\mathcal{D}_t$ or simply $\mathcal{D}_t$) for a $\GP(n,r,\mathcal{I},b)$ code operates as follows:
\begin{itemize}
\item For each $i\in\mathcal{I}$, compute $\hat{u}_i$ by applying the $\mathcal{D}_{t_i}$ decoder. The bits are successively decoded exactly in the same order as in the successive cancellation decoder of polar codes.
\item If $\hat{u}_i=\mathbf{e}$ for any $i\in\mathcal{I}$, stop decoding immediately and declare erasure.
\item If $\hat{u}_i\neq\mathbf{e}$ for every $i\in\mathcal{I}$, the output is $\hat{u}=(\hat{u}_i)_{i\in\mathcal{I}}$.
\end{itemize}
\end{mydef}

Two remarks are worth mentioning here:
\begin{itemize}
\item The computational complexity of any $\SCE$ decoder is $O(N\log N)$.
\item If $t_i=0$ for every $i\in\mathcal{I}$, we get a zero-undetected-error scheme.
\end{itemize}

\section{Erasure schemes using GP codes}

\begin{mydef}
Let $W:\mathbb{F}_2\longrightarrow\mathcal{Y}$ be a BMS channel and define
\begin{equation}
I_0^{\GP}(W):=\sum_{\substack{y\in\mathcal{Y}:\\W(y|1)=0}}W(y|0)=\sum_{\substack{y\in\mathcal{Y}:\\W(y|0)=0}}W(y|1).
\label{eqZUECap}
\end{equation}
It can be easily shown that $I_0^{\GP}(W)=p_W(0)$.
\end{mydef}

The following theorem, which is the main result of this paper, shows that $I_0^{\GP}(W)$ is the \emph{zero-undetected-error capacity of GP codes for $W$ under $\SCE$ decoders}.

\begin{mythe}
\label{mainThe}
Let $W$ be a fixed BMS channel. We have the following:
\begin{itemize}
\item For every $R<I_0^{\GP}(W)$, every $\beta<\frac{1}{2}$ and every $n$ large enough, there exists a GP code of blocklength $N=2^n$ and of rate at least $R$ for which the low complexity $\mathcal{D}_0\text{-}\SCE$ decoder (which induces a zero-undetected-error scheme) has an erasure probability of order $2^{-2^{\beta\cdot n}}$.
\item For every $\alpha>0$, every $\beta>\frac{1}{2}$, every $n$ large enough, and every GP code of rate $I_0^{\GP}(W)<R<I(W)$ and blocklength $N=2^n$, if $p_{er}<1-\alpha$ then $p_{ue}>2^{-2^{\beta\cdot n}}$. In other words, the undetected error probability cannot be made better than $2^{-N^{\frac{1}{2}+o(1)}}$ unless the erasure probability is of order $1-o(1)$.
\end{itemize}
\end{mythe}

In order to prove Theorem \ref{mainThe}, we need a few lemmas and propositions. The next proposition shows the first point of the Theorem. In fact, it provides a better estimate for the erasure probability:


\begin{myprop}
\label{prop1}
Let $W:\mathbb{F}_2\longrightarrow\mathcal{Y}$ be a BMS channel. For every $R<I_0^{\GP}(W)$, there exists a GP code of blocklength $N=2^n$ and of rate at least $R$ for which the low complexity $\mathcal{D}_0\text{-}\SCE$ decoder (which induces a zero-undetected-error scheme) has an erasure probability of order $2^{-2^{\frac{n}{2}+Q^{-1}\left(\frac{R}{I_0^{\GP}(W)}\right)\frac{\sqrt{n}}{2}+o(\sqrt{n})}}$, where $Q(x)=\mathbb{P}(\{\mathcal{N}(0,1)\geq x\})$ is the standard Q-function.
\end{myprop}
\begin{proof}
Define $W':\mathbb{F}_2\longrightarrow\mathbb{F}_2\cup\{\mathbf{e}\}$ as follows:
$$W'(y'|x)=\begin{cases} \displaystyle\sum_{\substack{y\in \mathcal{Y}:\\W(y|x\oplus 1)=0}} W(y|x)\quad&\text{if}\;y'=x,\\\vspace*{-4mm}
&\\
\displaystyle\sum_{\substack{y\in \mathcal{Y}:\\W(y|x\oplus 1)>0}} W(y|x)\quad&\text{if}\;y'=\mathbf{e},\\\vspace*{-4mm}
&\\
0\quad&\text{otherwise}. \end{cases}$$

In other words, for each $x\in\mathbb{F}_2$ we contract all the output symbols of $W$ for which we can decide without error that the input was $x$ to one output symbol of $W'$ that we also denote by $x$. Moreover, we contract all the remaining uncontracted symbols to the erasure symbol $\mathbf{e}$.

Let $\epsilon=1-I_0^{\GP}(W)$. One can easily check that $W'=\BEC(\epsilon)\preceq W$. Now for every $R<I_0^{\GP}(W)=1-\epsilon=I(W')$, there exists a polar code for $W'$ of rate at least $R$ and whose probability of error under successive cancellation decoder is equal to $2^{-2^{\frac{n}{2}+Q^{-1}\left(\frac{R}{I(W')}\right)\frac{\sqrt{n}}{2}+o(\sqrt{n})}}$ (see \cite{HasUrb}). One can use the same code for $W$ and apply the $\mathcal{D}_0\text{-}\SCE$ decoder. This induces a zero-undetected-error scheme.

It can be easily seen that the erasure probability for the $\mathcal{D}_0\text{-}\SCE$ decoder of the GP code for $W$ is of the same order as the error probability of the successive cancellation decoder of the polar code for $W'$.
\end{proof}

\vspace*{3mm}

In order to prove the second point of Theorem \ref{mainThe}, we will need the analysis tools of polarization theory. Let us first recall the basic notations and definitions.

Let $W:\mathbb{F}_2\longrightarrow\mathcal{Y}$ be a binary-input channel. We define the two channels $W^-:\mathbb{F}_2\longrightarrow\mathcal{Y}\times\mathcal{Y}$ and $W^+:\mathbb{F}_2\longrightarrow\mathcal{Y}\times\mathcal{Y}\times \mathbb{F}_2$ as follows:
\begin{equation}
W^-(y_1,y_2|u_1)=\frac{1}{2}\sum_{u_2\in \mathbb{F}_2}W(y_1|u_1\oplus u_2)W(y_2|u_2),
\label{eqMinus}
\end{equation}
\begin{equation}
W^+(y_1,y_2,u_1|u_2)=\frac{1}{2}W(y_1|u_1\oplus u_2)W(y_2|u_2).
\label{eqplus}
\end{equation}
For every $s=(s_1,\ldots,s_n)\in\{-,+\}^n$, we define $W^s$ recursively as $W^s:=((W^{s_1})^{s_2}\ldots)^{s_n}$.

\begin{myprop}
If $W$ is BMS, then $W^-$ and $W^+$ are BMS as well. More precisely, if $W\equiv\displaystyle\sum_{i=1}^l p_i\cdot\BSC(\epsilon_i)$ then
\begin{equation}
\label{eqWminus}
W^-\equiv \sum_{i=1}^l\sum_{j=1}^l p_ip_j\cdot\BSC(\epsilon_i\ast\epsilon_j),
\end{equation}
and
\begin{align}
W^+\equiv \sum_{i=1}^l\sum_{j=1}^l p_ip_j\cdot\Bigg(&(\epsilon_i\ast \epsilon_j)\cdot \BSC\left(\frac{\epsilon_i\overline{\epsilon}_j}{\epsilon_i\ast\epsilon_j}\right)+(\epsilon_i\ast \overline{\epsilon}_j)\cdot \BSC\left(\frac{\epsilon_i\epsilon_j}{\epsilon_i\ast\overline{\epsilon}_j}\right)\Bigg).\label{eqWplus}
\end{align}
\label{propEvo}
\end{myprop}
\begin{proof}
We use Equations \eqref{eqBMS}, \eqref{eqMinus} and \eqref{eqplus} and we apply the fact that $\BSC(\epsilon)\equiv\BSC(\overline{\epsilon})$ for every $\epsilon\in[0,1]$.
\end{proof}

\vspace*{3mm}

Proposition \ref{propEvo} can be used to derive the effect of polarization on $I_0^{\GP}(W)$ and $\epsilon_{\bic}(W)$ :

\begin{mycor}
\label{corEvoI}
$I_0^{\GP}(W^-)=I_0^{\GP}(W)^2$ and $I_0^{\GP}(W^+)=2I_0^{\GP}(W)-I_0^{\GP}(W)^2$.
\end{mycor}
\begin{proof}
Let $W\equiv\sum_{i=1}^l p_i\cdot \BSC(\epsilon_i)$ be a $\BSC$-decomposition of $W$. Using the equations of Proposition \ref{propEvo}, one can see that:
\begin{itemize}
\item $I_0^{\GP}(W^-)=p_{W^-}(0)\stackrel{(a)}{=}p_W(0)^2=I_0^{\GP}(W)^2$, where (a) follows from the fact that $m(\epsilon_i\ast\epsilon_j)=0$ if and only if $m(\epsilon_i)=m(\epsilon_j)=0$.
\item $I_0^{\GP}(W^+)=p_{W^+}(0)\stackrel{(b)}{=}2p_W(0)-p_W(0)^2=2I_0^{\GP}(W)-I_0^{\GP}(W)^2$, where (b) follows from the fact that
$$m\left(\frac{\epsilon_i\overline{\epsilon}_j}{\epsilon_i\ast\epsilon_j} \right)=0\quad\Leftrightarrow\quad m(\epsilon_i)=0\text{ or }m(\epsilon_j)=0,$$
and
$$m\left(\frac{\epsilon_i\epsilon_j}{\epsilon_i\ast\overline{\epsilon}_j}\right)=0\quad\Leftrightarrow\quad m(\epsilon_i)=0\text{ or }m(\epsilon_j)=0.$$
\end{itemize}
\end{proof}

\begin{mycor}
\label{corEvoE}
We have:
$$\epsilon_{\bic}(W^-)=\begin{cases}
2\epsilon_{\bic}(W)\cdot\overline{\epsilon_{\bic}(W)}\quad &\text{if }p_W(0)=0,\\
\epsilon_{\bic}(W)\quad &\text{otherwise}.
\end{cases}
$$
$$\epsilon_{\bic}(W^+)=\frac{\epsilon_{\bic}(W)^2}{\epsilon_{\bic}(W)^2+(1-\epsilon_{\bic}(W))^2}.$$
\end{mycor}
\begin{proof}
If $I(W)=1$ (i.e., $\epsilon_{\bic}(W)=0$), then $I(W^-)=I(W^+)=1$ which implies that $\epsilon_{\bic}(W^-)=\epsilon_{\bic}(W^+)=0$. This shows the corollary for $I(W)=1$.

Assume now that $I(W)<1$ so that $\epsilon_{\bic}(W)>0$. Let $\displaystyle W\equiv\sum_{i=1}^l p_i\cdot\BSC(\epsilon_i)$ be the canonical $\BSC$-decomposition of $W$.

Since $0\leq\epsilon_i,\epsilon_j\leq\frac{1}{2}$ for every $1\leq i,j\leq l$, it is easy to see that:
\begin{itemize}
\item $0\leq \epsilon_i\ast\epsilon_j\leq \frac{1}{2}$. This means that the crossover probabilities appearing in \eqref{eqWminus} do not need to be complemented.
\item $\epsilon_i\ast\epsilon_j=0$ if and only if $\epsilon_i=\epsilon_j=0$.
\end{itemize}
Now since the function $\epsilon \ast\epsilon'$ is increasing in both $\epsilon$ and $\epsilon'$ (assuming $0\leq \epsilon,\epsilon'\leq \frac{1}{2}$), we conclude that
\begin{align*}
\epsilon_{\bic}(W^-)&=\min_{\substack{1\leq i,j\leq l,\\ m(\epsilon_i\ast\epsilon_j)>0}}m(\epsilon_i\ast\epsilon_j)\\
&=\begin{cases}
2\epsilon_{\bic}(W)\cdot(1-\epsilon_{\bic}(W))\quad &\text{if }p_W(0)=0,\\
\epsilon_{\bic}(W)\quad &\text{otherwise}.
\end{cases}
\end{align*}

We apply a similar reasoning on $m\left(\frac{\epsilon_i\overline{\epsilon}_j}{\epsilon_i\ast\epsilon_j} \right)$ and $m\left(\frac{\epsilon_i\epsilon_j}{\epsilon_i\ast\overline{\epsilon}_j}\right)$. We obtain:
\begin{align*}
\epsilon_{\bic}(W^+)&=\min\Big\{\frac{\epsilon_i\epsilon_j}{\epsilon_i\ast\overline{\epsilon}_j},\frac{\epsilon_i\overline{\epsilon}_j}{\epsilon_i\ast\epsilon_j},1-\frac{\epsilon_i\overline{\epsilon}_j}{\epsilon_i\ast\epsilon_j}:\; 1\leq i,j\leq l,\;\epsilon_i>0,\;\epsilon_j>0\Big\}\\
&=\frac{\epsilon_{\bic}(W)^2}{\epsilon_{\bic}(W)^2+(1-\epsilon_{\bic}(W))^2}.
\end{align*}
\end{proof}

\begin{myprop}
\label{prop2}
Let $W:\mathbb{F}_2\longrightarrow\mathcal{Y}$ be a BMS channel and let $GP(n,r,\mathcal{I},b)$ be a generalized polar code of rate $R=\frac{r}{2^n}$ and blocklength $N=2^n$. If $I_0^{\GP}(W)<R<I(W)$ then for every $\beta>\frac{1}{2}$, every $\alpha>0$ and every $n$ large enough, there is no $\SCE$ decoder which can make the undetected error probability lower than $2^{-N^{\beta}}$ unless it makes the erasure probability at least $1-\alpha$.
\end{myprop}
\begin{proof}
Let $(B_n)_{n\geq1}$ be i.i.d. uniform random variables in $\{-,+\}$. Define the channel-valued process $(W_n)_{n\geq0}$ as follows:
\begin{align*}
W_0 &:= W,\\
W_{n} &:=W_{n-1}^{B_n}\;\forall n\geq1.
\end{align*}

Let $\frac{1}{2}<\beta'<\beta$ and let $n$ be large enough so that we have $\frac{\alpha}{2}\cdot 2^{-N^{\beta'}}\geq 2^{-N^{\beta}}$, where $N=2^n$.

Corollary \ref{corEvoI} shows that the process $I_0^{\GP}(W_n)$ is a martingale process. Therefore, $I_0^{\GP}(W_n)$ converges almost surely. Moreover, one can show by standard polarization theory techniques that $I_0^{\GP}(W_n)=p_{W_n}(0)$ converges almost surely to 0 or 1. Furthermore, for every $\epsilon>0$ we have:
$$\lim_{n\to\infty}\mathbb{P}(\{ p_{W_n}(0)<\epsilon \})=1-p_{W}(0).$$
Therefore, as $n$ becomes large, the fraction of indices $s\in\{-,+\}^n$ such that $p_{W^s}(0)\geq \epsilon$ is roughly at most $I^{\GP}_0(W)=p_W(0)$.

On the other hand, from Corollary \ref{corEvoE}, we can easily see that $\epsilon_{\bic}(W^-)\geq \epsilon_{\bic}(W)$ and $\epsilon_{\bic}(W^+)\geq \epsilon_{\bic}(W)^2$. By applying the same analysis of \cite{ArikanTelatar}, but to $\epsilon_{\bic}$ instead of the Bhattacharyya parameter, one can show that if $I(W)<1$, then the fraction of indices $s\in\{-,+\}^n$ such that $\epsilon_{\bic}(W^s)\geq 2^{-2^{\beta' n}}$ goes to 1. Therefore, for $n$ large enough, if $R>I_0^{\GP}(W)=p_W(0)$, there exists at least one index $s\in\{-,+\}^n$ whose corresponding index in $F^{\otimes n}$ appears in the generator matrix of the GP code and which satisfies $\epsilon_{\bic}(W^s)>2^{-2^{\beta' n}}$ and $p_{W^s}(0)<\frac{\alpha}{2}$. Let $i\in[2^n]$ be the index of the column of $F^{\otimes n}$ corresponding to $s$ and let $0\leq t_i\leq \frac{1}{2}$ be the threshold used for $W^s$ in an $\SCE-\mathcal{D}_t$ decoder. Let $p_{ue}^{(i)}$ and $p_{er}^{(i)}$ be the erasure probability and undetected error probability of the $\mathcal{D}_{t_i}$ decoder applied to $W^s$ respectively. We have:
$$\displaystyle p_{er}^{(i)}=\sum_{\epsilon>t_i}p_W(\epsilon),$$
and
\begin{align}
p_{ue}^{(i)}&=\sum_{\epsilon\leq t_i}\epsilon \cdot p_{W^s}(\epsilon)=\sum_{\epsilon_{\bic}(W^s)\leq\epsilon\leq t_i}\epsilon \cdot p_{W^s}(\epsilon)\nonumber\\
&\geq \sum_{\epsilon_{\bic}(W^s)\leq\epsilon\leq t_i}\epsilon_{\bic}(W^s) \cdot p_{W^s}(\epsilon)\nonumber\\
&=\epsilon_{\bic}(W^s)\cdot(1-p_{W^s}(0)-p_{er}^{(i)})\nonumber\\
&\geq 2^{-N^{\beta'}}\cdot\left(1-\frac{\alpha}{2}-p_{er}^{(i)}\right).\label{lalabatata}
\end{align}
Therefore, if $p_{er}^{(i)}\leq 1-\alpha$ then $p_{ue}^{(i)}\geq \frac{\alpha}{2}\cdot 2^{-N^{\beta'}}\geq 2^{-N^{\beta}}$. Hence $p_{ue}^{(i)}$ cannot be made less than $2^{-N^\beta}$ unless $p_{er}^{(i)}$ is at least $1-\alpha$. The proposition now follows from the fact that the erasure probability and the undetected error probability of the whole scheme are lower bounded by $p_{er}^{(i)}$ and $p_{ue}^{(i)}$ respectively.
\end{proof}

The proof of Theorem \ref{mainThe} now follows from Propositions \ref{prop1} and \ref{prop2}.

\section{Discussion}

The tradeoff obtained here between the undetected error probability and erasure probability for rates $R> I^{\GP}_0(W)$ is very sharp and does not depend on the rate $R$. A more refined estimation of the tradeoff between $p_{ue}$ and $p_{er}$ which is dependent on $R$ remains an open problem.

\section*{Acknowledgment}
I would like to thank Emre Telatar for helpful discussions.

\bibliographystyle{IEEEtran}
\bibliography{bibliofile}

\begin{thebibliography}{1}
\providecommand{\url}[1]{#1}
\csname url@samestyle\endcsname
\providecommand{\newblock}{\relax}
\providecommand{\bibinfo}[2]{#2}
\providecommand{\BIBentrySTDinterwordspacing}{\spaceskip=0pt\relax}
\providecommand{\BIBentryALTinterwordstretchfactor}{4}
\providecommand{\BIBentryALTinterwordspacing}{\spaceskip=\fontdimen2\font plus
\BIBentryALTinterwordstretchfactor\fontdimen3\font minus
  \fontdimen4\font\relax}
\providecommand{\BIBforeignlanguage}[2]{{%
\expandafter\ifx\csname l@#1\endcsname\relax
\typeout{** WARNING: IEEEtran.bst: No hyphenation pattern has been}%
\typeout{** loaded for the language `#1'. Using the pattern for}%
\typeout{** the default language instead.}%
\else
\language=\csname l@#1\endcsname
\fi
#2}}
\providecommand{\BIBdecl}{\relax}
\BIBdecl

\bibitem{Arikan}
E.~Ar{\i}kan, ``Channel polarization: A method for constructing
  capacity-achieving codes for symmetric binary-input memoryless channels,''
  \emph{Information Theory, IEEE Transactions on}, vol.~55, no.~7, pp. 3051
  --3073, 2009.

\bibitem{RMCodes}
\BIBentryALTinterwordspacing
S.~Kudekar, S.~Kumar, M.~Mondelli, H.~Pfister, E.~\c{S}a\c{s}o\u{g}lu, and
  R.~Urbanke, ``Reed-muller codes achieve capacity on erasure channels,''
  \emph{CoRR}, vol. abs/1601.04689, 2016. [Online]. Available:
  \url{http://arxiv.org/abs/1601.04689}
\BIBentrySTDinterwordspacing

\bibitem{ArikanTelatar}
E.~Ar{\i}kan and E.~Telatar, ``On the rate of channel polarization,'' in
  \emph{Information Theory, 2009. ISIT 2009. IEEE International Symposium on},
  28 2009.

\bibitem{HasUrb}
S.~Hassani, R.~Mori, T.~Tanaka, and R.~Urbanke, ``Rate-dependent analysis of
  the asymptotic behavior of channel polarization,'' vol.~59, no.~4, April
  2013, pp. 2267--2276.

\bibitem{TalVardy}
I.~Tal and A.~Vardy, ``List decoding of polar codes,'' in \emph{Information
  Theory Proceedings (ISIT), 2011 IEEE International Symposium on}, July 2011,
  pp. 1--5.

\bibitem{Forney}
J.~Forney, G.D., ``Exponential error bounds for erasure, list, and decision
  feedback schemes,'' \emph{Information Theory, IEEE Transactions on}, vol.~14,
  no.~2, pp. 206--220, Mar 1968.

\end{thebibliography}
\end{document}